\documentclass[aps,12pt,nofootinbib]{revtex4}

\usepackage{graphicx,graphics}
\usepackage{enumerate}
\usepackage{subfigure}
\usepackage{amsmath}
\usepackage{amsfonts}
\usepackage{amssymb}
\usepackage{amsthm}
\usepackage{psfrag}

\catcode`\@=11
\def\numberbysection{\@addtoreset{equation}{section}
\def\theequation{\arabic{section}.\arabic{equation}}}
\numberbysection

\theoremstyle{plain}
\newtheorem{theorem}{Theorem}
\newtheorem{lemma}{Lemma}

\newtheorem*{corollary}{Corollary}

\theoremstyle{definition}
\newtheorem{definition}{Definition}

\theoremstyle{remark}

\newcommand{\edge}[1]{\stackrel{#1}{\longrightarrow}}

\begin{document}

\title{Tighter Upper Bounds for the Minimum Number of Calls and Rigorous Minimal Time in Fault-Tolerant Gossip Schemes}

\author{
V.H.~Hovnanyan, H.E.~Nahapetyan, Su.S.~Poghosyan and V.S.~Poghosyan
}
\affiliation{
Institute for Informatics and Automation Problems\\ NAS of Armenia, 0014 Yerevan, Armenia
}

\begin{abstract}
  The gossip problem (telephone problem) is an information dissemination problem in
  which each of $n$ nodes of a communication network has a unique piece of information that must be transmitted
  to all the other nodes using two-way communications (telephone calls) between the pairs of nodes.
  During a call between the given two nodes, they exchange the whole information known to them at that moment.
  In this paper we investigate the $k$-fault-tolerant gossip problem, which is a generalization of the gossip problem,
  where at most $k$ arbitrary faults of calls are allowed.
  The problem is to find the minimal number of calls $\tau(n,k)$ needed to guarantee the $k$-fault-tolerance.
  We construct two classes of $k$-fault-tolerant gossip schemes (sequences of calls) and found two upper bounds of $\tau(n,k)$,
  which improve the previously known results.
  The first upper bound for general even $n$ is $\tau(n,k) \leq \frac{1}{2} n \left\lceil\log_2 n\right\rceil + \frac{1}{2} n k$.
  This result is used to obtain the upper bound for general odd $n$.
  From the expressions for the second upper bound it follows that $\tau(n,k) \leq \frac{2}{3} n k + O(n)$ for large $n$.
  Assuming that the calls can take place simultaneously, it is also of interest to find $k$-fault-tolerant gossip schemes, which can spread the full
  information in minimal time. For even $n$ we showed that the minimal time is $T(n,k)=\lceil\log_2 n\rceil + k$.
\end{abstract}

\maketitle


\section{Introduction}

The gossiping is one of the basic problems of information dissemination in communication networks.
The gossip problem (also known as a telephone problem) is attributed to A. Boyd, although to the best knowledge of the reviewers,
it was first formulated by R. Chesters and S. Silverman (Univ. of Witwatersrand, unpublished, 1970).
Consider a set of $n$ persons (nodes) each of which initially knows some unique piece of information
that is unknown to the others, and they can make a sequence of telephone calls to spread the information.
During a call between the given two nodes, they exchange the whole information known to them at that moment.
The problem is to find a sequence of calls with minimum length (minimal gossip scheme), by which all the nodes will know all pieces of the information.
It has been shown in numerous works \cite{BakerShostak,Bumby,HanjalMilnerSzemeredi,Tijdeman} that the minimal number of calls
is $2n-4$ when $n \geq 4$ and $0,1,3$ for $n=1,2,3$, respectively.
Since then many variations of gossip problem have been introduced and investigated
(see ex. \cite{Seress,West,FertinRespaud, Labahn1, Labahn2, FertinLabahn}).

One of the natural generalizations of this problem is the $k$-fault-tolerant gossip problem, which
assumes that some of the calls in the call sequence can fail (do not take place) \cite{HoShigeno,BermanHawrylycz, HaddadRoySchaffer,HasunumaNagamochi}.
The nodes cannot change the sequence of their future calls depending on the current failed calls.
Here the aim is to find a minimal gossip scheme, which guarantees the full exchange of the information in the case of at most $k$ arbitrary fails,
regardless of which the calls failed.
The gossip schemes, which provide $k$-fault-tolerance, are called $k$-fault-tolerant gossip schemes.
Denote the minimal number of calls in the $k$-fault-tolerant minimal gossip scheme by $\tau(n,k)$.

Berman and Hawrylycz \cite{BermanHawrylycz} obtained the lower and upper bounds for $\tau(n,k)$:
\begin{equation}
\left\lceil \left( \frac{k+4}{2} \right) \left( n-1 \right) \right\rceil
- 2 \left\lceil \sqrt{n} \right\rceil + 1
\leq \tau(n,k) \leq
\left\lfloor \left( k + \frac{3}{2} \right) \left( n-1 \right) \right\rfloor
\end{equation}
for $k \leq n-2$, and
\begin{equation}
\left\lceil \left( \frac{k+3}{2} \right) \left( n-1 \right) \right\rceil
- 2 \left\lceil \sqrt{n} \right\rceil
\leq \tau(n,k) \leq
\left\lfloor \left( k + \frac{3}{2} \right) \left( n-1 \right) \right\rfloor
\end{equation}
for $k \geq n-2$.

Afterwards, Hadded, Roy and Schaffer \cite{HaddadRoySchaffer} proved that
\begin{equation}
\tau(n,k) \leq \left( \frac{k}{2} + 2p \right) \left( n-1 + \frac{n-1}{2^p-1} + 2^p \right),
\end{equation}
where $p$ is any integer between $1$ and $\log_2 n$ inclusive.
By choosing $p$ appropriately, this result improves the upper bounds obtained by Berman and Hawrylycz for almost all $k$.
Particularly, by choosing $p = \left[\frac{\log_2 n}{2}\right]$, the following bound is obtained: $\tau(n,k) \leq \frac{nk}{2} + O(k\sqrt{n} + n \log_2 n)$.

For the special case when $n = 2^p$ for some integer $p$, Haddad, Roy, and Schaffer \cite{HaddadRoySchaffer} also showed that
\begin{eqnarray}
\tau(n,k) &\leq& \min \left\{
\left( \left\lceil \frac{k+1}{\log_2 n} \right\rceil + 1 \right)
\frac{n \log_2 n}{2},\right.\nonumber\\
&~&\left.
\left( \left\lfloor \frac{k+1}{\log_2 n} \right\rfloor + 1 \right)
\frac{n \log_2 n}{2} +
\left( \left(k+1\right) \bmod \log_2 n \right) \left( 2 n - 4 \right)
\right\}.
\end{eqnarray}
Thus, $\tau(n,k) \leq \frac{nk}{2} + O(n \log_2 n)$, when $n$ is a power of 2.

Later on, Hou and Shigeno \cite{HoShigeno} showed that
\begin{equation}
\left\lfloor \frac{n(k+2)}{2} \right\rfloor \leq \tau(n,k) \leq \frac{n(n-1)}{2} + \left\lceil\frac{nk}{2}\right\rceil.
\end{equation}
Thus, it holds that $\frac{nk}{2} + \Omega(n) \leq \tau(n,k) \leq \frac{nk}{2} + O(n^2)$.
These bounds improve the previous bounds for small $n$ and sufficiently large $k$.

Recently, Hasunuma and Nagamochi \cite{HasunumaNagamochi} showed that
\begin{equation}
\tau(n,k) \leq \left\{
\begin{array}{c}
 \frac{n\log_2 n}{2} + \frac{nk}{2}, \quad \textrm{if $n$ is a power of 2} \\
 2 n \left\lfloor \log_2 n \right\rfloor + n \left\lceil \frac{k-1}{2} \right\rceil, \quad \textrm{otherwise,}
\end{array}
\right.
\end{equation}
and
\begin{equation}
\left\lceil \frac{3n - 5}{2} \right\rceil +
\left\lceil \frac{1}{2} \left( n k + \left\lfloor \frac{n + 1}{2} \right\rfloor - \left\lfloor \log_2 n \right\rfloor \right) \right\rceil
\leq \tau(n,k).
\end{equation}
From their results, it holds that
$\tau(n,k) \leq \frac{nk}{2} + O(n \log_2 n)$. Particularly, their upper bound improves the upper bound by Hou and Shigeno for all $n \geq 13$.
They also improve the upper bound by Haddad {\it et al.} by showing that the factor $(k/2+2p)$
in their upper bound can be replaced with a smaller factor $(k/2 + p)$:
\begin{equation}
\tau(n,k) \leq \left( \frac{k}{2} + p \right) \left( n-1 + \frac{n-1}{2^p-1} + 2^p \right),
\end{equation}
where $p$ is any integer between $1$ and $\log_2 n$ inclusive.

In this paper we construct two classes of $k$-fault-tolerant gossip schemes based on Kn\"odel graphs \cite{FertinRespaud}
and a wheel graph, which improve the previously known results on the upper bound for the number of calls.
The obtained expressions for general $n$ and $k$ are
(see Theorems \ref{theorem-Knedel}, \ref{theorem-wheel-odd}, \ref{theorem-wheel-even}).
\begin{equation}
\tau (n, k) \leq \left\{
\begin{array}{ll}
\frac{n}{2}   \left\lceil \log_2 n \right\rceil       + \frac{nk}{2}             , & \hbox{for even $n$,} \\
\frac{n-1}{2} \left\lceil \log_2 {(n-1)} \right\rceil + \frac{(n-1)k}{2} + 2(k+1), & \hbox{for odd $n$,}
\end{array}
\right.
\end{equation}
and
\begin{equation}
\tau(n,k) \leq \left\{
  \begin{array}{ll}
  \frac{2}{3}(n-1)k     + \frac{5}{2}(n-1), & \hbox{if\;\; $(k \bmod 3) = 0$}, \\
  \frac{2}{3}(n-1)(k-1) + \frac{7}{2}(n-1), & \hbox{if\;\; $(k \bmod 3) = 1$}, \\
  \frac{2}{3}(n-1)(k-2) +           4(n-1), & \hbox{if\;\; $(k \bmod 3) = 2$}, \\
  \end{array}
\right.
\end{equation}
for odd $n$,
\begin{equation}
\tau(n,k) \leq \left\{
   \begin{array}{ll}
   \frac{1}{3}(2n-1)k     + \frac{5}{2}n - 4, & \hbox{if\;\; $(k \bmod 3) = 0$}, \\
   \frac{1}{3}(2n-1)(k-1) + \frac{7}{2}n - 5, & \hbox{if\;\; $(k \bmod 3) = 1$}, \\
   \frac{1}{3}(2n-1)(k-2) +          4 n - 5, & \hbox{if\;\; $(k \bmod 3) = 2$}. \\
   \end{array}
\right.
\end{equation}
for even $n$.
Particularly, for large $n$ we have
\begin{equation}
\tau(n,k) \leq \frac{2}{3}\; n k + O(n).
\end{equation}

Assuming that the calls between non-overlapping pairs of nodes can take place simultaneously,
it is also of interest to find $k$-fault-tolerant gossip schemes, which spread the full information in minimal time.
For even $n$ we showed that the minimal time is
\begin{equation}
T(n,k)=\lceil\log_2 n\rceil + k.
\end{equation}

\section{Preliminaries}

A gossip scheme (a sequence of calls between $n$ nodes) can be represented by an undirected edge-labeled graph $G=(V,E)$ with $n$ vertices
($|V| \equiv |V(G)| = n$).
The vertices and edges of $G$ represent correspondingly the nodes and the calls between the pairs of nodes of a gossip scheme.
Such graphs may have multiple edges, but not self loops.
An edge-labeling of $G$ is a mapping $t_G:E(G) \rightarrow \mathbb{Z}^{+}$.
The label $t_G(e)$ of the given edge $e\in E(G)$ represents the moment of the time, when the corresponding call is occurred.

A sequence $P = (v_0, e_1, v_1, e_2, v_2, \ldots , e_k, v_k)$ with vertices $v_i \in V (G)$ for $0 \leq i \leq k$ and
edges $e_i \in E(G)$ for $1 \leq i \leq k$ is called a walk from a vertex $v_0$ to a vertex $v_k$ in $G$ with length $k$, if
each edge $e_i$ joins two vertices $v_{i-1}$ and $v_i$ for $1 \leq i \leq k$.
A walk, in which all the vertices are distinct is called a path.
If $t_G(e_i) < t_G(e_j)$ for $1 \leq i < j \leq k$, then $P$ is an ascending path from $v_0$ to $v_k$ in $G$.
Given two vertices $u$ and $v$, if there is an ascending path from $u$ to $v$, then $v$ receives the information of $u$.
Note that two different edges can have the same label.
Since we consider only (strictly) ascending paths, then such edges (i.e. calls) are independent,
which means that the edges with the same label can be reordered arbitrarily but for any $t_1 < t_2$ all the edges with the label $t_1$ are ordered
before any of the edges with the label $t_2$.

\begin{definition}
The communication between two vertices of $G$ is called $k$-failure safe if an ascending path between them
remains, even if arbitrary $k$ edges of $G$ are deleted (the corresponding calls are failed).
The graph $G$ is called a $k$-fault-tolerant gossip graph if the communication between all the pairs of its vertices is $k$-failure safe.
\end{definition}

From Menger theorem \cite{Menger} it follows that a $k$-fault-tolerant gossip graph is a graph whose edges are labeled
in such a way that there are at least $k+1$ edge-disjoint ascending paths between two arbitrary vertices.
A $0$-fault-tolerant gossip graph is simply called a gossip graph.

To describe the construction of $k$-fault-tolerant gossip graphs (schemes), we use some important definitions and propositions
given in the works \cite{HaddadRoySchaffer,HasunumaNagamochi}.
First of all, in order to simplify the discussion for edge-disjoint paths, we often omit the vertices (or edges)
in the description of a path if there is no confusion.

\begin{definition}
Let $P = (e_1, e_2, \ldots , e_k)$ be a path with edges $e_i \in E(G)$ for $1 \leq i \leq k$ in a labeled graph $G$.
If $P$ is divided into $s + 1$ subpaths $P^{(1)} = (e_1, \ldots , e_{p_1} )$, $P^{(2)} = (e_{p_1 + 1}, \ldots , e_{p_2} )$, $\,\ldots\,$,
$P^{(s+1)} = (e_{p_s + 1}, \ldots , e_k)$, then we write $P = P^{(1)} \odot P^{(2)} \odot \cdots \odot P^{(s+1)}$,
where $\odot$ is the concatenation operation on two paths for which the last vertex of one path is the first vertex of the other.
If $P = P^{(1)} \odot P^{(2)} \odot \cdots \odot P^{(s+1)}$ such that $P^{(j)}$ is an ascending path for $1 \leq j \leq s+1$ and $P^{(j)} \odot P^{(j+1)}$
is not an ascending path for $1 \leq j \leq s$, then $P$ is an $s$-folded ascending path in $G$. For an $s$-folded ascending path $P$, the
folded number of $P$ is defined to be $s$.
\end{definition}

\begin{definition}
Consider two graphs $G_1=(V, E_1)$ and $G_2=(V, E_2)$ with the same set of vertices $V$
and labeled edge sets $E_1$ and $E_2$, respectively.
The edge sum of these graphs is a graph $G_1 + G_2 = G = (V, E)$ with $E = E_1\cup E_2$, whose edges $e \in E$ are labeled by the following rules:
\begin{equation}
t_G(e)= \left\{
        \begin{array}{ll}
         t_{G_1}(e) , & \hbox{if $e\in E_1$}, \\
         t_{G_2}(e)+ \max\limits_{e'\in E_1} t_{G_1}(e'), & \hbox{if $e\in E_2$}.
        \end{array}
      \right.
\end{equation}
The edge sum $G_1 + G_2 + \cdots + G_h$ of $h$ identical graphs $(G_1 = G_2 = \cdots = G_h \equiv G)$ is denoted by $hG$.
Each set $E(G_i)$ in $hG$ is denoted by $E_i(hG)$, i.e. $E(hG) = \bigcup_{1 \leq i \leq h} E_i(hG)$.
Note that the labels of the edges in $E_i(hG)$ are greater than the corresponding edges in $E(G)$ by $(i-1) \times \max\limits_{e \in E(G)} t_{G}(e)$.
Given a subset of edges $A \subseteq E(G)$, denote its copy in the set $E_i(hG)$ by $A_i$.
By this analogy, a path $P$ in $G$ as a subset of $E(G)$ has a copy in $E_i(hG)$, which we denote by $P_i$.
\end{definition}

Let $P = P^{(1)} \odot P^{(2)} \odot \cdots \odot P^{(s+1)}$ be an $s$-folded ascending path from a vertex $u$ to a vertex $v$ in $G$,
where $P^{(j)}$ is an ascending subpath for $1 \leq j \leq s+1$.
Then, $P_i$ is also an $s$-folded ascending path and $P_i = P_i^{(1)} \odot P_i^{(2)} \odot \cdots \odot P_i^{(s+1)}$ for $1 \leq i \leq h$.
Now consider the path $P(k) = P_{k}^{(1)} \odot P_{k+1}^{(2)} \odot \cdots \odot P_{k+s}^{(s+1)}$ in $hG$.
Then, $P(k)$ is an ascending path from $u$ to $v$ for $1 \leq k \leq h - s$ such that $P(k)$ and $P(k\,')$ are edge-disjoint if $k \neq k\,'$.
Thus, based on $P$, we can construct $(h - s)$ edge-disjoint ascending paths from $u$ to $v$ in $hG$.
Similarly, based on another $s$-folded ascending path $P\,'$ from $u$ to $v$, we can construct $(h-s)$ edge-disjoint ascending
paths $P\,'(k)$ from $u$ to $v$ for $1 \leq k \leq h-s$. If $P$ and $P\,'$ are edge-disjoint, then all the paths $P(1), \ldots , P(h - s)$
and $P\,'(1), \ldots , P\,'(h - s)$ are pairwise edge-disjoint by construction.
Therefore, the following lemma holds (see the works \cite{HaddadRoySchaffer} and \cite{HasunumaNagamochi}).
\begin{lemma}
Let $u$ and $v$ be vertices in a labeled graph $G$.
If there are $p$ edge-disjoint $s$-folded ascending paths from $u$ to $v$ in $G$, then there are $p(h - s)$ edge-disjoint ascending
paths from $u$ to $v$ in $hG$ for any integer $h \geq s$.
\end{lemma}
From this lemma, if there are $p$ edge-disjoint $s$-folded ascending paths from $u$ to $v$ in $G$,
then there are $k + 1$ edge-disjoint ascending paths from $u$ to $v$ in $\left( s + \lceil \frac{k+1}{p} \rceil \right)G$.
Thus, the following corollary is obtained (see \cite{HaddadRoySchaffer}).
\begin{corollary}
Let $G$ be a graph with $n$ vertices and $m$ edges. If there are $p$ edge-disjoint $s$-folded ascending paths between every pair
of vertices in a labeled graph $G$, then $\tau(n,k) \leq \left( s + \lceil \frac{k+1}{p} \rceil \right)m$.
\end{corollary}
In order to improve this estimation of the upper bound, a stronger proposition is formulated and proved in \cite{HasunumaNagamochi}.

\begin{theorem}
Let $G$ be a labeled graph with $n$ vertices. Suppose that
\begin{itemize}
  \item $E(G)$ can be decomposed into $l$ subsets $F^{(0)},F^{(1)},\ldots,F^{(l-1)}$ such that for any two edges $e \in F^{(i)}$
        and $e\,' \in F^{(j)}$, $t_G(e) < t_G(e\,')$ if $i < j$,
  \item for any two vertices $u$ and $v$, there are $p$ edge-disjoint paths from $u$ to $v$ such that the sum
        of their folded numbers is at most $q$, and the last edges of $r_i$ paths are in $F^{(i)}$ for $0 \leq i \leq l-1$.
\end{itemize}
Then, the minimal number of edges in a $k$-fault-tolerant gossip graph is bounded
\begin{equation}
\tau(n,k) \leq \xi(n,k),
\end{equation}
with $\xi(n,k)$ defined by the expression
\begin{equation}
\xi(n,k) = \sum_{0 \leq i \leq w} |F^{(i \bmod l)}|,
\end{equation}
where $w$ is an integer satisfying
\begin{equation}
\sum\limits_{0 \leq i \leq w} r_{i \bmod l} \geq k+q+1.
\end{equation}
\label{theorem-1}
\end{theorem}
During the proof, the graph
$\widetilde{G} = hG + G\,'$ with $h = \lfloor \frac{w}{l} \rfloor$ and $G\,'=(V,\cup_{ 0 \leq i \leq w - h l} F^{(i)})$
is constructed, and showed that it is a $k$-fault-tolerant gossip graph.
The number of edges of this graph is $|E(\widetilde{G})| = \sum_{0 \leq i \leq w} |F^{(i \bmod l)}|$.

In the next two sections we construct two classes of labeled graphs and apply Theorem \ref{theorem-1} to improve the known estimations
of the upper bounds for $\tau(n,k)$.

\section{The $k$-fault-tolerant gossip graphs based on Kn\"odel graphs}
\label{section-Knodel}

The family of Kn\"odel graphs \cite{FertinRespaud} is defined as follows:

\begin{definition}
The Kn\"odel graph on $n\geq2$ vertices ($n$ even) and of degree $\Delta \geq 1$ is denoted by $W_{\Delta,n}$.
The vertices of $W_{\Delta,n}$ are the pairs $(i, j)$ with $i=1,2$ and $0 \leq j \leq n/2 - 1$.
For every $j$, $0 \leq j \leq n/2 - 1$ and $l=1,\ldots ,\Delta$,
there is an edge with the label $l$ between the vertex $(1, j)$ and $(2, (j+ 2^{l-1} -1) \bmod n/2)$.
\end{definition}

An example of Kn\"odel graph is shown in Fig.\ref{Fig-Knodel}.
Note that $W_{1,n}$ consists of $n/2$ disconnected edges.
For $\Delta \geq 2$, $W_{\Delta,n}$ is connected.
\begin{figure}[!ht]
  \includegraphics[width=120mm]{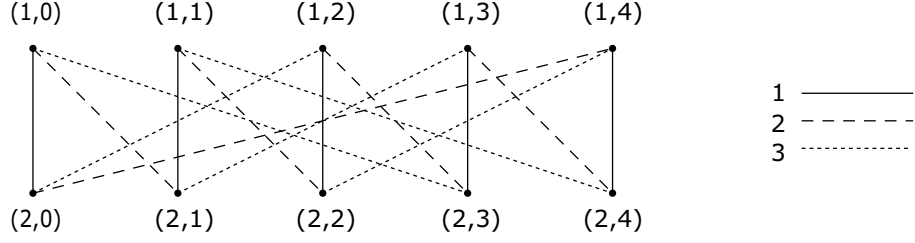}
  \caption{The Kn\"odel graph $W_{3,10}$ with a number of vertices $n=10$ and degree $\Delta=3$.
  The solid, dashed and dotted edges are labeled correspondingly 1, 2 and 3.}
  \label{Fig-Knodel}
\end{figure}

\begin{definition}
The interval between two vertices $(\alpha,\beta)$ and $(\gamma,\delta)$ in the Kn\"odel graph $W_{\lfloor \log_2 n \rfloor,n}$ is
defined to be
\begin{equation}
R(\,(\alpha,\beta);\;(\gamma,\delta)\,) =
   \left\{
        \begin{array}{ll}
          \delta-\beta,               & \hbox{if $\delta \geq \beta$ and $\alpha=1$,} \\
          \frac{n}{2}-|\delta-\beta|, & \hbox{if $\delta < \beta$ and $\alpha=1$,} \\
          |\delta-\beta|,             & \hbox{if $\delta \leq \beta$ and $\alpha=2$,} \\
          \frac{n}{2}-(\delta-\beta), & \hbox{if $\delta > \beta$ and $\alpha=2$.}
        \end{array}
      \right.
\end{equation}
\end{definition}

Given two arbitrary vertices $(\alpha,\beta)$ and $(\gamma,\delta)$ in the Kn\"odel graph $W_{\lfloor \log_2 n \rfloor,n}$,
the ascending path from vertex $(\gamma,\delta)$ to $(\alpha,\beta)$ exists only if
\begin{equation}
R(\,(\alpha,\beta);\;(\gamma,\delta)\,) \leq 2^{\lfloor \log_2 n \rfloor - 1} - 1.
\label{Rmax-gossip}
\end{equation}
The explicit expression for the ascending path from vertex $(\gamma,\delta)$ to $(\alpha,\beta)$ is
described via three sequences $\{a_i\}$, $\{b_i\}$ and $\{f_i\}$ defined recursively by the following way:
\begin{equation}
a_1 = \left\{
        \begin{array}{ll}
          1, & \hbox{if $\alpha=2$,} \\
          2, & \hbox{if $\alpha=1$,}
        \end{array}
      \right.
\quad\quad\quad
a_i = \left\{
        \begin{array}{ll}
          1, & \hbox{if $a_{i-1}=2$,} \\
          2, & \hbox{if $a_{i-1}=1$,}
        \end{array}
      \right. \label{a_i}
\end{equation}

\begin{equation}
f_1 = \left\{
        \begin{array}{ll}
          2^{\left\lceil \log_2 (\delta-\beta+1) \right\rceil} -1,               & \hbox{if $\delta \geq \beta$ and $\alpha=1$,} \\
          2^{\left\lceil \log_2 (\frac{n}{2}-|\delta-\beta|+1) \right\rceil} -1, & \hbox{if $\delta < \beta$ and $\alpha=1$,} \\
          2^{\left\lceil \log_2 (|\delta-\beta|+1) \right\rceil} -1,             & \hbox{if $\delta \leq \beta$ and $\alpha=2$,} \\
          2^{\left\lceil \log_2 (\frac{n}{2}-(\delta-\beta)+1) \right\rceil} -1, & \hbox{if $\delta > \beta$ and $\alpha=2$.}
        \end{array}
      \right.
\end{equation}

\begin{equation}
b_1 = \left\{
        \begin{array}{ll}
          (\beta+f_1) \bmod n/2,             & \hbox{if $\alpha=1$,} \\
          (\frac{n}{2}+\beta-f_1) \bmod n/2, & \hbox{if $\alpha=2$.}
        \end{array}
      \right.
\end{equation}

\begin{equation}
f_i = \left\{
        \begin{array}{ll}
          2 ^{\left\lceil \log_2 (|b_{i-1}-\delta|+1) \right\rceil} -1,               & \hbox{if $b_{i-1} < \delta$ and $a_{i-1}=1$,} \\
          2^{\left\lceil \log_2 (\frac{n}{2}-(b_{i-1}-\delta)+1) \right\rceil} -1, & \hbox{if $b_{i-1} \geq \delta$ and $a_{i-1}=1$,} \\
          2 ^{\left\lceil \log_2 (b_{i-1}-\delta+1) \right\rceil} -1,             & \hbox{if $b_{i-1} \geq \delta$ and $a_{i-1}=2$,} \\
          2^{\left\lceil \log_2 (\frac{n}{2}-|b_{i-1}-\delta|+1) \right\rceil} -1, & \hbox{if $b_{i-1} < \delta$ and $a_{i-1}=2$.}
        \end{array}
      \right.
\end{equation}

\begin{equation}
b_i = \left\{
        \begin{array}{ll}
          (\beta+\sum_{j=1}^i {(-1)^{j-1}f_j}) \bmod n/2,             & \hbox{if $\alpha=1$;} \\
          (\frac{n}{2}+\beta-\sum_{j=1}^i {(-1)^{j-1}f_j}) \bmod n/2, & \hbox{if $\alpha=2$.}
        \end{array}
      \right. \label{b_i}
\end{equation}

The sequences $\{a_i\}$, $\{b_i\}$ and $\{f_i\}$ stop on the minimal index $i=L$, for which $a_L=\gamma$ and $b_L=\delta$.
Then $(\;(a_L,b_L),\;(a_{L-1},b_{L-1}),\ldots,\; (a_1,b_1),\;(\alpha,\beta)\;)$ is
the ascending path from vertex $(\gamma,\delta)$ to the vertex $(\alpha,\beta)$.
For example, the ascending path from the vertex $(2,2)$ to $(1,0)$ in Fig.\ref{Fig-Knodel} is
\begin{equation}
(2,2) \edge{1} (1,2) \edge{2} (2,3) \edge{3} (1,0),
\end{equation}
where the numbers on arrows are the labels of the corresponding edges.
\begin{figure}[!ht]
  \includegraphics[width=120mm]{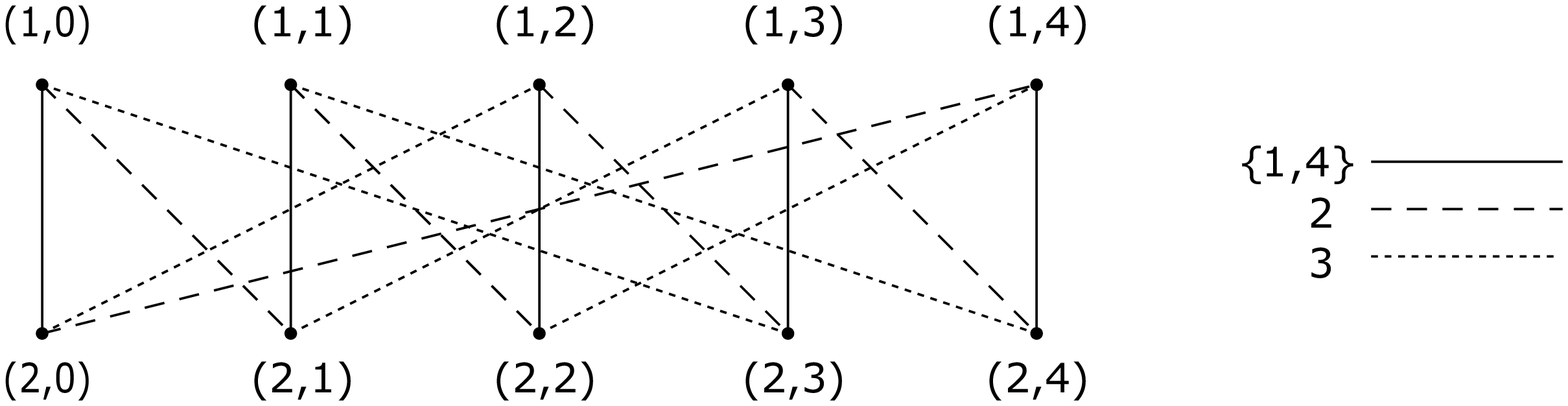}
  \caption{The edge sum of two graphs: $W_{3,10} + W_{1,1}$.
  Each of the vertical (solid) lines is a two edges with labels 1 and 4.
  The dashed and dotted edges are labeled correspondingly 2 and 3.
  This graph is a gossip graph with $n=10$ nodes ($0$-fault-tolerant gossip graph based on Kn\"odel graphs).}
  \label{Fig-gossip_n=10}
\end{figure}

In Fig.\ref{Fig-gossip_n=10} the edge sum of two Kn\"odel graphs $W_{3,10} + W_{1,10}$ is shown.

\begin{lemma}
The graph $W_{\lfloor \log_2 n \rfloor,\;n} + W_{1,n}$ is a gossip graph ($n$ is even).
Moreover, if $n$ is a power of 2, then $W_{\log_2 n,\;n}$ is already a gossip graph.
\end{lemma}

\begin{proof}
Consider the graph $G = W_{\Delta,n} + W_{1,n}$ with any $n$ and $\Delta$.
Let us fix an arbitrary vertex $(1,\beta) \in G$, $0 \leq \beta \leq n/2-1$.
The set of vertices from which there are ascending paths in $W_{\Delta,n}$ to the vertex $(1,\beta)$ is
\begin{equation}
V_{(1,\beta)} = \left\{(\gamma, (\beta + \delta) \bmod n/2): \; \gamma = 1,2; \; \delta = 0,1,\ldots, 2^{\Delta - 1} - 1 \right\}.
\label{V1beta}
\end{equation}
Similarly, the set of vertices from which there are ascending paths in $W_{\Delta,n}$ to a vertex $(2,\beta) \in G$, $0 \leq \beta \leq n/2-1$ is
\begin{equation}
V_{(2,\beta)} = \left\{(\gamma, (\beta - \delta + n/2) \bmod n/2): \; \gamma=1,2; \; \delta = 0,1,\ldots, 2^{\Delta - 1} - 1 \right\}.
\label{V2beta}
\end{equation}

If $n$ is a power of 2 $(n=2^m)$, for $\Delta = m \equiv \log_2 n$ we have $V_{(1,\beta)} = V_{(2,\beta)} = V(W_{\Delta,n})$.
Therefore, $W_{\log_2 n,n}$ is a gossip graph.

For the case, when $n \neq 2^m$, the addition of the graph $W_{1,n}$ to the graph $W_{\Delta,n}$ connects the vertices $(1,\beta)$ and $(2,\beta)$
with a new edge with a label $\Delta+1$.
This label is bigger than the labels in $W_{\Delta,n}$ (by construction), therefore this edge (call) exchanges
the full information of the nodes $(1,\beta)$ and $(2,\beta)$ collected during the calls in $W_{\Delta,n}$.
Therefore, after all calls the pieces of information known by nodes $(1,\beta)$ and $(2,\beta)$ coincide and
are nothing but the following union of sets:
\begin{equation}
\widetilde{V}_{\beta} = V_{(1,\beta)} \cup V_{(2,\beta)}\,.
\end{equation}
Note that, for $\Delta = \lfloor \log_2 n \rfloor$, $\widetilde{V}_{\beta} = V(G)$ for any $\beta$.
Therefore, $W_{\lfloor \log_2 n \rfloor,n} + W_{1,n}$ is a gossip graph.
\end{proof}

Consider the problem to find the minimal number of calls $\tau(n,k)$ in a $k$-fault-tolerant gossip scheme with $n$ nodes.
For the special case, when $n$ is a power of 2 ($n = 2^m$) Hasunuma and Nagamochi \cite{HasunumaNagamochi}
constructed $k$-fault-tolerant gossip graphs based on hypercubes and found
\begin{equation}
\tau (n, k) \leq \frac{n}{2} \log_2 n + \frac{nk}{2}\,.
\label{tau-Nagamochi}
\end{equation}
Now we will generalize this result for general $n$.
The following lemma is an important step towards achieving this goal.
Some details of the proof are skipped to make the text not too overloaded and formal,
but all important properties and all possible cases are considered in detail.

\begin{lemma}
Consider two arbitrary vertices $(\alpha,\beta)$ and $(\gamma,\delta)$ in the Kn\"odel graph $W_{\lfloor \log_2 n \rfloor,n}$
with $\alpha,\gamma = 1,2$; $\beta,\delta = 0,1,2,\ldots,n/2-1$.
There are $p = \lfloor \log_2 n \rfloor$ edge-disjoint folded ascending paths from the vertex $(\gamma,\delta)$ to the vertex $(\alpha,\beta)$
and the sum of their folded numbers is at most $q = \lfloor \log_2 n \rfloor$.
\label{lemma-folded-paths}
\end{lemma}

\begin{proof}
Since the graph $W_{\lfloor \log_2 n \rfloor,n}$ is symmetric, then without loss of generality it is enough to consider the vertex $(\alpha,\beta) = (1,0)$.
Consider the following pairwise non-intersecting sets of vertices
\begin{equation}
V(1) = \left\{ (2, 0) \right\},
\end{equation}
\begin{equation}
V(\Delta) = \left\{(i, 2^{\Delta - 2} + j): \; i = 1,2; \; j =0,1,\ldots,\, 2^{\Delta - 2} - 1 \right\},\quad \Delta=2,3,\ldots,\lfloor \log_2 n \rfloor.
\end{equation}

From (\ref{Rmax-gossip}) it follows that there exist an ascending path from every vertex in $V(\Delta)$, $1 \leq \Delta \leq \lfloor \log_2 n \rfloor$
to the vertex $(1,0)$, which is described by Eqs. (\ref{a_i}) -- (\ref{b_i}).
Moreover, all endpoints of the edges of this path except the last vertex of the last edge, which is the vertex $(1,0)$, are in $V(\Delta)$.

Consider the case $n=2^m-2$.
It can be verified that there exists an ascending path or a $1$-folded ascending path from the vertex $(\gamma,\delta)$ to the one of the vertices in
$V(\Delta)$ for any $1 \leq \Delta \leq \lfloor \log_2 n \rfloor$.

Therefore, there are $\lfloor \log_2 n \rfloor$ folded ascending paths from $(\gamma,\delta)$ to $(1,0)$ with folded numbers that can be $0$, $1$ or $2$.
It can be proved that these paths are edge-disjoint.

Depending on the given vertex $(\gamma, \delta)$ there are $2$ possible cases:

{\it Case 1}: $R(\,(1,0);\;(\gamma,\delta)\,) > \lfloor \log_2 n \rfloor$.
In this case, all paths are $1$-folded. Each of the paths consists of two parts, which are themselves ascending paths.
The first path starts from $(\gamma,\delta)$ and comes to one of the vertices of the sets $V(\Delta)$,  $(1 \leq \Delta \leq \lfloor \log_2 n \rfloor)$.
The second path goes from that vertex up to the vertex $(1,0)$. Thus the sum of folded numbers of all paths is $\lfloor \log_2 n \rfloor$.

{\it Case 2}: $R(\,(1,0);\;(\gamma,\delta)\,) \leq \lfloor \log_2 n \rfloor$. In this case, there is an ascending path from $(\gamma,\delta)$ to $(1,0)$.
On the other hand, there is an ascending path started from the given vertex $(\gamma, \delta)$, whose single edge has a label $\lfloor \log_2 n \rfloor$
and finishes outside of any $V(\Delta)$,  $(1 \leq \Delta \leq \lfloor \log_2 n \rfloor)$.
Therefore, the folding number of a folded ascending path containing this path is $2$.
The remaining folded ascending paths from $(\gamma,\delta)$ to $(1,0)$ are similar as in the Case $1$ and are 1-folded.
Therefore, the sum of folded numbers of all folded ascending paths from $(\gamma,\delta)$ to $(1,0)$ is $\lfloor \log_2 n \rfloor$.

For $n\neq 2^m-2$ the sum of folded numbers of all folded ascending paths from some vertices
$(\gamma,\delta)$ to $(1,0)$ is $\lfloor \log_2 n \rfloor - 1$, but for these graphs also we have $q=\lfloor \log_2 n \rfloor$.
\end{proof}

\begin{theorem}
The minimal number of calls (edges) in a $k$-fault-tolerant gossip graphs satisfies the inequality
\begin{equation}
\tau (n, k) \leq \left\{
\begin{array}{ll}
\frac{n}{2}   \left\lceil \log_2 n \right\rceil       + \frac{nk}{2}             , & \hbox{for even $n$,} \\
\frac{n-1}{2} \left\lceil \log_2 {(n-1)} \right\rceil + \frac{(n-1)k}{2} + 2(k+1), & \hbox{for odd $n$.}
\end{array}
\right.
\label{tau-knodel}
\end{equation}
\label{theorem-Knedel}
\end{theorem}

\begin{proof}
From (\ref{tau-Nagamochi}) it follows that it is enough to consider the case, when $n$ is not a power of 2 $(n \neq 2^m)$.
First, assume that $n$ is even.
To construct a $k$-fault-tolerant gossip graph we simply take the graph $G = W_{\lfloor \log_2 n \rfloor,n}$ and apply Theorem \ref{theorem-1}.
The edge set $E(G)$ of $G$ is divided into the following $l=\lfloor \log_2 n \rfloor$ subsets:
\begin{equation}
E(G) = F^{(0)} \cup F^{(1)} \cup F^{(2)} \ldots \cup F^{(\lfloor \log_2 n \rfloor - 1)},
\end{equation}
where
\begin{equation}
F^{(i)} = \left\{e: \;\; e \in G; \;\; t_G(e) = i+1 \right\}; \quad i = 0, 1, 2, \ldots, \lfloor \log_2 n \rfloor - 1.
\end{equation}
The number of edges in $F^{(i)}$ is $|F^{(i)}| = n/2$ for any $i = 0, 1, 2, \ldots, \lfloor \log_2 n \rfloor - 1$.
Each vertex in $G$ has $\lfloor \log_2 n \rfloor$ incident edges with one edge in each of the sets
$F^{(i)}$ for $i = 0, 1, 2, \ldots, \lfloor \log_2 n \rfloor - 1$.
In Lemma \ref{lemma-folded-paths} it is shown that for any two vertices $(\alpha,\beta)$ and $(\gamma,\delta)$ with
$\alpha,\beta = 1,2$; $\gamma,\delta = 0,1,2,\ldots,n/2-1$, there are $p = \lfloor \log_2 n \rfloor$ edge-disjoint folded ascending paths from the vertex $(\gamma,\delta)$ to the vertex $(\alpha,\beta)$ and the sum of their folded numbers is at most $q = \lfloor \log_2 n \rfloor$.
Therefore, $r_i=1$ for any $i = 0, 1, 2, \ldots, \lfloor \log_2 n \rfloor - 1$, and $w$ is any
integer satisfying the inequality $w \geq k+q$.
From this result, the minimal value of $\xi(n,k)$ is
\begin{equation}
\xi(n,k) = \frac{n}{2}(w+1) = \frac{n}{2} \left\lceil \log_2 n \right\rceil + \frac{nk}{2},
\end{equation}
which proves the theorem for even $n$.

For the case when $n$ is odd, we take the above constructed $k$-fault-tolerant gossip graph with $n-1$ vertices with one additional vertex $v$.
Denote this graph $G\,'=(V,E\,')$. Now consider a graph $G\,''=(V,E\,'')$ with vertex set $V$ and edge set $E\,''$ consisting of one edge,
which connects the vertex $v$ to one of the remaining $n-1$ vertices, say the vertex $(1,0)$. Consider the graph
\begin{equation}
G = (k+1)G\,'' + G\,' + (k+1)G\,''.
\end{equation}
Here, the vertices $v$ and $(1,0)$ make $k+1$ calls, then the vertices in $V\symbol{`\\} \{v\}$ make calls corresponding to $G\,'$,
and then the vertices $v$ and $(1,0)$ make $k+1$ additional calls. It is easy to see that $G$ is a $k$-fault-tolerant gossip graph.
Therefore, for odd $n$, we have $\tau (n,k)\leq\tau (n-1,k)+2(k+1)$, from which we obtain
\begin{equation}
\tau (n, k) \leq \frac{n-1}{2} \left\lceil \log_2 {(n-1)} \right\rceil + \frac{(n-1)k}{2} + 2(k+1),
\label{tau-knodel-odd}
\end{equation}

\end{proof}

\section{Construction of $k$-fault-tolerant gossip graphs based on wheel graphs}

Consider a wheel graph $G=(V,E)$ with an odd number of vertices $n=2k+1$, whose vertices and edges are labeled by the following rules.
The label of the central vertex is $u$.
The remaining $2k$ vertices (which are located on the circle) are labeled consequentially $v_1,v_1',v_2,v_2',\ldots,v_k,v_k'$.
Since the periodic boundary conditions are assumed, we identify the vertices $v_{i \pm k} \equiv v_{i}$ and $v_{i \pm k}' \equiv v_{i}'$ for $i=1,2,\ldots,k$.
The set of edges consists of three subsets
\begin{equation}
E(G) = F^{(0)} \cup F^{(1)} \cup F^{(2)}
\end{equation}
with
\begin{eqnarray}
F^{(0)}  &=& \left\{(v_i,v_i'):\;\; t_G((v_i,v_i')) = 1,\;\; i=1,2,\ldots,k \right\},\\
F^{(1a)} &=& \left\{(v_i',u):\;\; t_G((v_i',u)) = 2,\;\; i=1,2,\ldots,k \right\},\\
F^{(1b)} &=& \left\{(v_i,u):\;\; t_G((v_i,u)) = 3,\;\; i=1,2,\ldots,k \right\},\\
F^{(1)}  &=& F^{(1a)} \cup F^{(1b)},\\
F^{(2)}  &=& \left\{(v_i',v_{i+1}):\;\; t_G((v_i',v_{i+1})) = 4,\;\; i=1,2,\ldots,k \right\}.
\end{eqnarray}
In Fig.\ref{Fig-wheel-1} the wheel graph $G$ for $n=11$ vertices is shown.
\begin{figure}[!ht]
  \includegraphics[width=80mm]{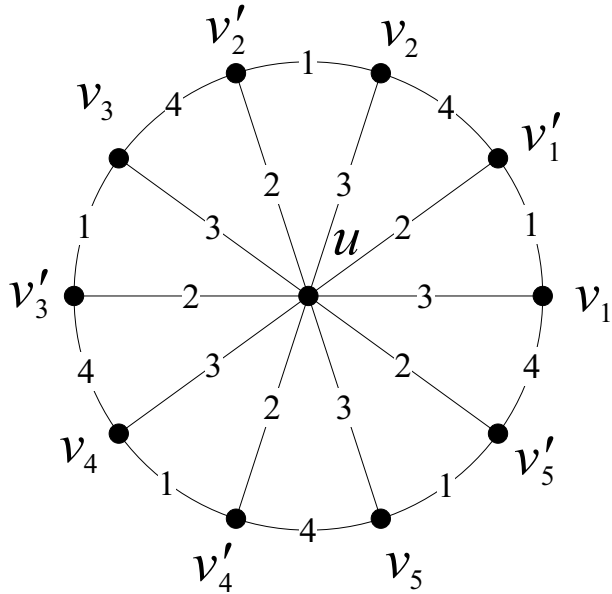}
  \caption{Wheel graph for odd $n$ (here $n=11$).}
  \label{Fig-wheel-1}
\end{figure}
\begin{figure}[!ht]
  \includegraphics[width=80mm]{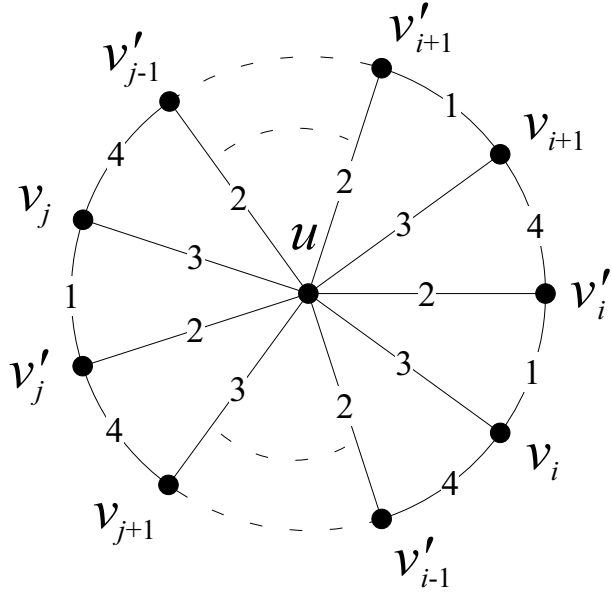}
  \caption{The illustration of arbitrarily fixed vertices $v_i$, $v'_i$, $v_j$, $v'_j$ and their neighbourhood in the wheel graph.}
  \label{Fig-wheel-2}
\end{figure}
Now we are going to apply Theorem \ref{theorem-1} to this graph.
First, for all the pairs of vertices in $G$, we construct $3$ edge-disjoint folded ascending paths from the first vertex to the second vertex.
For $i,j=1,2,\ldots,k$ and $j \neq i-1,\;i,\;i+1,\;i+2$ (see Fig.\ref{Fig-wheel-2} for illustration) we have
\begin{itemize}
  \item from $v_i$ to $u$
  \begin{itemize}
    \item $ v_i \edge{3} u $
    \item $ v_i \edge{1} v'_i \edge{2} u $
    \item $ v_i \edge{4} v'_{i-1} \edge{2} u $
  \end{itemize}

  \item from $u$ to $v'_i$
  \begin{itemize}
    \item $ u \edge{2} v'_i $
    \item $ u \edge{3} v_{i} \edge{1} v'_i $
    \item $ u \edge{3} v_{i+1} \edge{4} v'_i $
  \end{itemize}

  \item from $v_i$ to $v_j$
  \begin{itemize}
    \item $ v_i \edge{3} u \edge{2} v'_{j-1} \edge{4} v_j $
    \item $ v_i \edge{1} v'_i \edge{2} u \edge{3} v_{j+1} \edge{4} v'_j \edge{1} v_j $
    \item $ v_i \edge{4} v'_{i-1} \edge{2} u \edge{3} v_j $
  \end{itemize}

  \item from $v_i$ to $v'_j$
  \begin{itemize}
    \item $ v_i \edge{3} u \edge{2} v'_j $
    \item $ v_i \edge{1} v'_i \edge{2} u \edge{3} v_{j} \edge{1} v'_j $
    \item $ v_i \edge{4} v'_{i-1} \edge{2} u \edge{3} v_{j+1} \edge{4} v'_j $
  \end{itemize}

  \item from $v'_i$ to $v_j$
  \begin{itemize}
    \item $ v'_i \edge{2} u \edge{3} v_{j+1} \edge{4} v'_j \edge{1} v_j $
    \item $ v'_i \edge{1} v_i \edge{3} u \edge{2} v'_{j-1} \edge{4} v_j $
    \item $ v'_i \edge{4} v_{i+1} \edge{1} v'_{i+1} \edge{2} u \edge{3} v_j $
  \end{itemize}

  \item from $v'_i$ to $v'_j$
  \begin{itemize}
    \item $ v'_i \edge{2} u \edge{3} v_j \edge{1} v'_j$
    \item $ v'_i \edge{1} v_i \edge{3} u \edge{2} v'_j $
    \item $ v'_i \edge{4} v_{i+1} \edge{1} v'_{i+1} \edge{2} u \edge{3} v_{j+1} \edge{4} v'_j $
  \end{itemize}

\end{itemize}
The edge-disjoint folded ascending paths for $j = i-1,\;i,\;i+1,\;i+2$ are shorter.
They have less or equal folded numbers, and are easier to construct.
Therefore we do not give them to avoid the artificial growth of the text.
Finally, we have
\begin{eqnarray}
|F^{(0)}| = (n-1)/2, && |F^{(1)}| = n-1,\quad |F^{(2)}| = (n-1)/2,\\
p=3, && r_0=r_1=r_2=1,\quad q=3,
\end{eqnarray}
from which we obtain $w \geq k+3$ and the minimal value of $\xi(n,k)$ is
\begin{equation}
\xi(n,k) =
\left\{
  \begin{array}{ll}
  \frac{2}{3}(n-1)k     + \frac{5}{2}(n-1), & \hbox{if\;\; $(k \bmod 3) = 0$}, \\
  \frac{2}{3}(n-1)(k-1) + \frac{7}{2}(n-1), & \hbox{if\;\; $(k \bmod 3) = 1$}, \\
  \frac{2}{3}(n-1)(k-2) +           4(n-1), & \hbox{if\;\; $(k \bmod 3) = 2$}. \\
  \end{array}
\right.
\label{xi-wheel-odd}
\end{equation}
Therefore, from Theorem \ref{theorem-1} the following theorem holds.
\begin{theorem}
The minimal number of calls $\tau (n, k)$ in a $k$-fault-tolerant gossip graph with an odd number of vertices $n$ satisfies the inequality
$\tau(n, k) \leq \xi(n, k)$, where $\xi(n, k)$ is defined by (\ref{xi-wheel-odd}).
\label{theorem-wheel-odd}
\end{theorem}

For even $n$, we modify the wheel graph by adding a new vertex $u'$ and transforming the edge set to the following expression
\begin{equation}
E(G) = F^{(0)} \cup F^{(1)} \cup F^{(2)}
\end{equation}
\begin{figure}[!ht]
  \includegraphics[width=120mm]{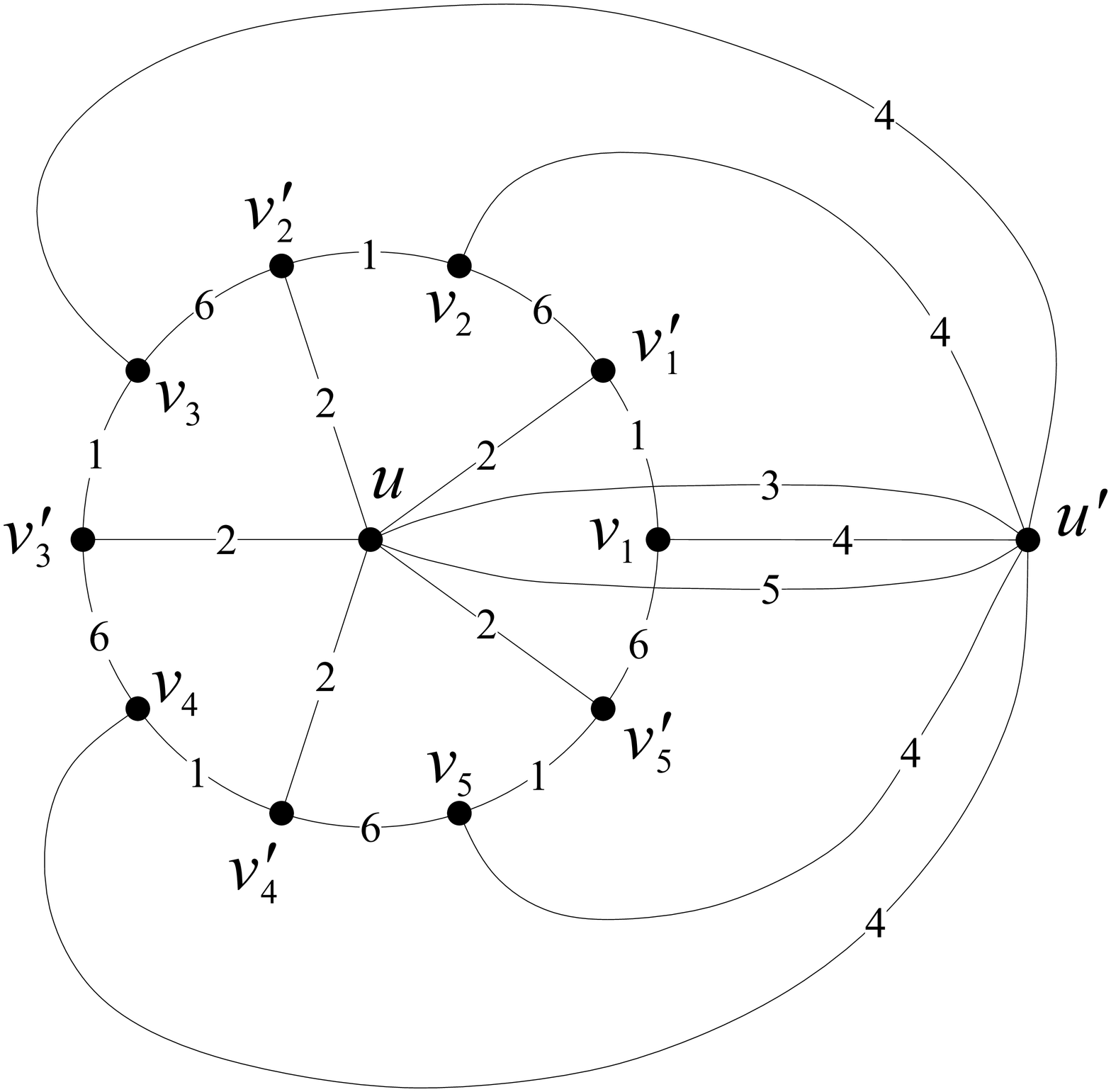}
  \caption{Wheel graph for even $n$ (here $n=12$).}
  \label{Fig-wheel-3}
\end{figure}
with
\begin{eqnarray}
F^{(0)}  &=& \left\{(v_i,v_i'):\;\; t_G((v_i,v_i')) = 1,\;\; i=1,2,\ldots,k \right\},\\
F^{(1a)} &=& \left\{(v_i',u):\;\; t_G((v_i',u)) = 2,\;\; i=1,2,\ldots,k \right\},\\
e_a &=& (u,u'),\;\; t_G(e_a) = 3,\\
F^{(1b)} &=& \left\{(v_i,u'):\;\; t_G((v_i,u')) = 4,\;\; i=1,2,\ldots,k \right\},\\
e_b &=& (u,u'),\;\; t_G(e_b) = 5,\\
F^{(1)}  &=& F^{(1a)} \cup F^{(1b)} \cup \{ e_a, e_b \},\\
F^{(2)}  &=& \left\{(v_i',v_{i+1}):\;\; t_G((v_i',v_{i+1})) = 6,\;\; i=1,2,\ldots,k \right\}.
\end{eqnarray}
Here the vertices $u$ and $u'$ are connected by two edges $e_a$ and $e_b$.
The graph $G$ for $n=12$ vertices is shown in Fig.\ref{Fig-wheel-3}.
Constructing the edge-disjoint folded ascending paths, one obtains
\begin{eqnarray}
|F^{(0)}| = (n-2)/2, && |F^{(1)}| = n,\quad |F^{(2)}| = (n-2)/2,\\
p=3, && r_0=r_1=r_2=1,\quad q=3,
\end{eqnarray}
which give $w \geq k+3$ and the minimal value of $\xi(n,k)$ is
\begin{equation}
\xi(n,k) =
\left\{
   \begin{array}{ll}
   \frac{1}{3}(2n-1)k     + \frac{5}{2}n - 4, & \hbox{if\;\; $(k \bmod 3) = 0$}, \\
   \frac{1}{3}(2n-1)(k-1) + \frac{7}{2}n - 5, & \hbox{if\;\; $(k \bmod 3) = 1$}, \\
   \frac{1}{3}(2n-1)(k-2) +          4 n - 5, & \hbox{if\;\; $(k \bmod 3) = 2$}. \\
   \end{array}
\right.
\label{xi-wheel-even}
\end{equation}
Therefore, from Theorem \ref{theorem-1} the following theorem holds.
\begin{theorem}
The minimal number of calls $\tau (n, k)$ in a $k$-fault-tolerant gossip graph with an even number of vertices $n$ satisfies the inequality
$\tau(n, k) \leq \xi(n, k)$, where $\xi(n, k)$ is defined by (\ref{xi-wheel-even}).
\label{theorem-wheel-even}
\end{theorem}
From Theorem \ref{theorem-wheel-odd} (for odd $n$) and Theorem \ref{theorem-wheel-even} (for even $n$) we have the asymptotical expansion of the upper bound:
\begin{equation}
\tau(n,k) \leq \frac{2}{3}\; n k + O(n).
\end{equation}

At the end, considering the special cases $k = 1$ and $k = 2$, for which the construction of $k$-fault-tolerant gossip schemes becomes simpler,
after some modifications of the wheel graph, one can demonstrate some $k$-fault-tolerant gossip schemes,
which give slightly better formulas comparing with above obtained results.
Without giving the details of the definitions of such schemes, we give the obtained upper bounds:
\begin{equation}
\tau (n, 1) \leq 2n-3+\left\lfloor\frac{n}{2}\right\rfloor, \quad \tau (n, 2) \leq 3n - 3.
\end{equation}

\section{fault-tolerant gossiping in minimum time}
\label{section-time}

Now consider the following problem.
What is the minimum time needed to complete $k$-fault-tolerant gossip scheme if the non-overlapping pairs of nodes
can make calls simultaneously ?

Haddad {\it et al.} \cite{HaddadRoySchaffer} considered this problem, and obtained the following upper bounds:
\begin{eqnarray}
T(n,k) \leq m+\left\lceil \frac{k+1}{m} \right\rceil m &\quad& \hbox{if $n=2^m$; \;\; $r > 0$},\\
T(n,k) \leq 8m+2m\left\lceil \frac{k+1}{\left\lceil \frac{m}{2} \right\rceil} \right\rceil &\quad& \hbox{if $n=2^{m-1}+r<2^m$; \;\; $r > 0$}.
\end{eqnarray}
The tighter upper and lower bounds are obtained by Gargano \cite{Gargano}:
\begin{eqnarray}
T(n,k) \leq m+k    &\quad& \hbox{if $n = 2^m$}, \label{Tmax} \\
T(n,k) \leq m+3k+1 &\quad& \hbox{if $n = 2^{m-1} + r < 2^m$; \;\; $r > 0$},\\
\notag\\
T(n,k) \geq m+k    &\quad& \hbox{if $n=2^{m-1} + r \leq 2^m$; \;\; $r > 0$}.
\label{Tmin}
\end{eqnarray}
Note that here $m = \lceil \log_2 n \rceil$.

Our method of construction of $k$-fault-tolerant gossip graphs based on Kn\"odel graphs (see section \ref{section-Knodel}) with Eqs. (\ref{Tmax}), (\ref{Tmin}) allows us
to obtain the exact value of the minimal time for even $n$.
Since the labels of the graphs are integer valued and the Kn\"odel graphs have no two incident edges with the same label,
we can interpret the labels as the moments of discrete time, when the current call takes place.
Therefore, the calls with the same label are independent and take place simultaneously.
Given a labeled graph $G$, the maximal label of $G$ is the duration of time to compete all calls.
Since the maximal label of the $k$-fault-tolerant gossip graph based on Kn\"odel graph with an even number of vertices $n$ and $n \neq 2^m$ is
$\left\lceil \log_2 n \right\rceil+k$, then
\begin{equation}
T(n,k) \leq \left\lceil \log_2 n \right\rceil + k \quad \hbox{ if $n$ is even and $n \neq 2^m$ }.
\label{T-Knodel}
\end{equation}
Combining Eqs. (\ref{Tmax}), (\ref{Tmin}) and (\ref{T-Knodel}), we obtain the exact value of $T(n,k)$ for even $n$:
\begin{equation}
T(n,k) = \left\lceil \log_2 n \right\rceil + k \quad \hbox{if $n$ is even}.
\label{T-exact-even}
\end{equation}

\end{document}